\documentclass[a4paper,preprint,sort&compress,12pt]{amsart}
\usepackage[left=3.4 cm,top=2.4cm,bottom=3.4cm,right=3.4cm]{geometry}
\usepackage{enumerate}
\usepackage{xspace}
\usepackage{amsfonts,amssymb,amsthm,amsmath}
\usepackage{bbm}
\usepackage{xcolor}

\definecolor{green(munsell)}{rgb}{0.0, 0.66, 0.47}
\definecolor{BlueGreenn}{rgb}{0.3,0.5,0.8}
\definecolor{DB}{rgb}{0.3,0.3,0.3}
\definecolor{DOr}{rgb}{0.7,0.3,0.3}
\definecolor{DGr}{rgb}{0.3,0.65,0.3}
\definecolor{DBl}{rgb}{0.1,0.3,0.5}
\definecolor{arylideyellow}{rgb}{0.91, 0.84, 0.42}
\definecolor{burntorange}{rgb}{0.8, 0.33, 0.0}
\definecolor{chromeyellow}{rgb}{1.0, 0.65, 0.0}
\usepackage[hypertexnames=true, citecolor = burntorange, colorlinks = true, linkcolor = BlueGreenn, pdffitwindow = false, urlbordercolor = white,pdfborder={0 0 0}]{hyperref}%
\usepackage{databib}
\usepackage{graphicx}
\usepackage{placeins}

\usepackage[T1]{fontenc}
\usepackage{color}
\usepackage[utf8]{inputenc}
\usepackage{framed}
\usepackage{centernot}
\usepackage{tikz-cd} 
\usepackage{tkz-euclide}
\usepackage{mathtools}
\usepackage{subcaption}
\usepackage{ragged2e}
\usepackage{blindtext}
\usepackage{sidecap, caption}

\numberwithin{equation}{section}

\newtheorem{theorem}{Theorem}[section]

\newtheorem{lemma}[theorem]{Lemma}

\newtheorem{remark}[theorem]{Remark}

\newcommand{\be}{\begin{equation}}
\newcommand{\ee}{\end{equation}}
\def\bea{\begin{eqnarray}}
\def\eea{\end{eqnarray}}

\usepackage{mathtools, mathrsfs}

\usepackage{ytableau} 
\ytableausetup{boxsize=0.5em,centertableaux}



\usepackage{tikz}

\begin{document}

\title{Vanishing Cycles and Analysis of Singularities \\
of Feynman~Diagrams}

\author{Stanislav Srednyak}
\address{Department of Physics, Duke University, Durham, NC 27708, USA\\ stanislav.srednyak@duke.edu}

\author{Vladimir Khachatryan}
\address{Department of Physics, Duke University, Durham, NC 27708, USA\\vladimir.khachatryan@duke.edu}
\address{Department of Physics, Indiana University, Bloomington, IN 47405, USA\\vlakhach@indiana.edu}
\date{July 20, 2023}

\begin{abstract}
In this work, we analyze the vanishing cycles of Feynman loop integrals by the means of the Mayer--Vietoris spectral sequence. 
A complete classification of possible vanishing geometries is obtained. We use this result for establishing an asymptotic 
expansion for the loop integrals near their singularity locus and then give explicit formulas for the coefficients of such an expansion. 
Further development of this framework may potentially lead to exact calculations of one- and two-loop Feynman 
diagrams, as well as other next-to-leading and higher-order diagrams, in studies of radiative corrections for upcoming lepton--hadron 
scattering experiments.
\end{abstract}
	 
\maketitle

\tableofcontents

\section{Introduction}\label{section:intro}

The mathematical study of Feynman integrals has a long history~\cite{fotiadi1965applications,pham2011singularities,leray1959calcul}. 
In particular, Refs.~\cite{pham2011singularities,leray1959calcul}  pioneered homological methods, and~the use of Picard--Lefschetz theory for 
Feynman integrals has influenced the development of the theory of hyperfunctions~\cite{kashiwara1976b}. The~book~\cite{hwa1966homology} contains 
reprints of many important papers. The~decomposition of the homology of the complement of the union of several hypersurfaces has a long history by 
itself. Although~it is addressed in classical works, it still attracts attention~\cite{kontsevich1992intersection} in the context of derived algebraic 
geometry. The~book~\cite{eden2002analytic} contains more focused applications of mathematical methods to scattering processes, paying 
particular attention to the geometry of spaces of momenta as dictated by the combinatorics of Feynman quadrics. Most importantly for our 
development discussed in the next sections, the~use of integration-by-parts (IBP) relations was discussed in~\cite{griffiths1968periods} 
(an earlier work is~\cite{manin1958algebraic}).

In addition to  classical work, a~major development was initiated in the mathematical literature~\cite{GZK}. This development concerned 
the use of hypergeometric functions. This approach is very conceptually attractive and has gained a number of followers in the physics literature~\cite{ananthanarayan2023feyngkz,Fevola:2023fzn,delaCruz:2019skx,Feng:2019bdx,muhlbauer2022cutkosky,Pathak:2023nvs,YelleshpurSrikant:2019khx}. 
An earlier work of one author clarified the relationship to flat bundles. The~use of differential equations  is the subject of 
many recent works~\cite{dlapa2020deriving,smirnov2020choose}. The~analysis of singularities and  jumps at the singularities is the subject 
of modern work~\cite{hannesdottir2022implications,Bourjaily:2020wvq}.

Calculations of Feynman loop integrals in particular provide an important test for some modern mathematical methods of quantum 
field theory and particle physics in general, including  relativistic scattering theory. The~full non-perturbative 
theory is a much more complicated problem that incorporates perturbative Feynman integrals as its small part. 
Therefore, it is essential to develop new mathematical methods that can help achieve detailed understanding of perturbative 
integrals. Additionally, there is also ample theoretical and phenomenological interest coming from studies of various topics, such as the 
anomalous magnetic moment of the electron ($g$-2) \cite{aoyama2012tenth,aoyama2020anomalous,jegerlehner2008anomalous,keshavarzi2018muon}, 
 two-loop corrections to the Higgs production~\cite{campbell2016associated,Hessenberger:2022tcx,Ahmed:2021hrf}, and one- and two-loop corrections to lepton--proton~\cite{Maximon:2000hm,Gramolin:2014pva,AGIM:2015,Bucoveanu:2018soy,Fadin:2018dwp,Banerjee:2020rww,Afanasev:2021nhy,Kaiser:2022pso}
and lepton--lepton scattering~\cite{AGIM:2015,Banerjee:2020rww,Shumeiko:1999zd,Kaiser:2010zz,Aleksejevs:2013gxa,Aleksejevs:2020vxy,Zykunov:2021fxh,Banerjee:2021qvi,Bondarenko:2022kxq}
processes.

There has hitherto been much progress in using algebro-geometric techniques for analyses of Feynman integrals. 
The topology of contours and its implications have been analyzed in~\cite{Frellesvig:2020qot,Frellesvig:2019uqt}. 
The flat connection in terms of master integrals has been analyzed in a number of cases; see, e.g.,~
\cite{DiVita:2018nnh,bonisch2021analytic,remiddi2014schouten,Chaubey:2019lum}. In~addition, the~connection to 
modular forms has been explored in~\cite{Pogel:2022yat,Weinzierl:2020fyx}. The emergence of non-polylogarithmic 
functions has been studied in~\cite{Bloch:2014qca,Bloch:2016izu}.  There has been a lot of work on this problem 
in the physics literature~\cite{delaCruz:2024xit,Pogel:2022vat,Klausen:2019hrg,Pogel:2022ken}.

With regard to  perturbative integrals, they define a class of functions for external momenta and masses of 
interacting/scattering physical particles. These are vector functions because multiple integration chains 
must be considered for them. In~fact, they are sections of a flat vector bundle, the~pullback of the Gauss--Manin 
connection to the physical locus. Such functions have regular singularities along  Landau varieties and~
are also the pullback of the solutions to the holonomic GZK $D$-module~\cite{gel1989hypergeometric}. Pullback 
here is meant in the sense of the theory of holonomic $D$-modules~\cite{saito2013grobner,hibi2017pfaffian}. Such 
pullback can be taken on a singular set in the parameter space. It is known that physical spaces correspond to a 
deep stratum in the space of generic kinematic invariants~\cite{Fevola:2023fzn,srednyak2018universal}.  
The general theory of such pullbacks involves the theory of $b$-functions~\cite{kashiwara2003d}. For~a 
general reference on holonomic $D$-modules and associated Gauss--Manin connections, see~\cite{hotta2007d,gelfand1994discriminants,hibi2017pfaffian}. For~the role of Landau varieties in the 
analysis of integrals of the hypergeometric type (in particular, from~the point of view of underlying 
meromorphic connection), see~\cite{arnold2012singularities,kashiwara2016regular}.

Although a large body of general mathematical knowledge has been accumulated about the aforementioned objects~\cite{arnold2012singularities,deligne2006equations,saito2013grobner}, some elementary questions about them 
remain unaddressed. For~example, first-order equations, which those functions satisfy to, are still missing, 
such that these equations might be consistent with the symmetry of Feynman diagrams. While there has been 
a lot of research about differential equations' methods (emerging as dominant methods of computation) 
\cite{delaCruz:2024xit,delaCruz:2019skx}, there is no closed-form expression for those equations. This is especially 
true for  first-order equations, or~the Pfaffian form~\cite{hibi2017pfaffian}. Given this, such equations 
could be very important for the numerical evaluation of  Feynman loop integrals. Their state of the art can be 
assessed by focusing on the family of propagator-type diagrams~\cite{bonisch2021analytic,adams2014two,Pogel:2022yat}. 
For the evaluation of the loop integrals, one must start with an analysis of their behavior near singularities. 
Although particular cases, such as double-logarithmic asymptotic analysis, are well known, the~general case has not 
been treated in the literature. Furthermore, the~exponents of  asymptotic expansion near the singularities should 
provide information about the diagonal part of the flat connection, when it is restricted to the singularity 
locus. Information on the connection~matrices is very limited. 

Just as a summary, the~study of singularities of perturbative Feynman loop integrals is an important subject, and~from 
a mathematical point of view, Feynman integrals are examples of integrals of rational functions with parameters. 
The theory of their singularities is very complex, and~most results, including those of this work, focus on special cases.
In particular, our paper focuses on an asymptotic analysis of  perturbative loop integrals near their singularities, 
with some essential applications discussed in the paper. The~structure of this paper is organized as~follows.

\begin{itemize}

\item ({Section}
~\ref{section:prelim}) We start the analysis in the momentum $q$-space.

\item  (Section~\ref{section:cycle}) Then, we continue with a discussion of vanishing cycles since they are a mechanism 
that leads to the formation of singularities. One can observe the usefulness of the Mayer--Vietoris spectral 
sequence~for a classification of possible types of singularities ({for 
 a general introduction to the use of  Mayer--Vietoris sequences as applied to the topology 
of projective complete intersections, we refer to Refs.~\cite{dimca1992singularities,spanier1989algebraic}}). We give a complete classification of the localization of 
the vanishing cycles with the case of generic polynomials. Non-local vanishing cycles are outlined as~well.

\item  (Section~\ref{section:pinch}) Herein, we introduce the so-called ``pinch map'': a map from the singular locus 
to the space of the (virtual) loop momenta of particles that gives the location of a point to which the vanishing 
cycle becomes contractible. The~pinch map simplifies calculations for  asymptotic analysis and~must have some meaning 
for generic GZK functions as well. Two cases of pinching are presented: with general polynomials and with Feynman 
loop~integrals.

\item  (Section~\ref{section:exam}) We consider specific examples of vertex loop diagrams, which may potentially
be applied to the lepton--proton ($l-p$) scattering process (the scattering of spin-1/2 particles).

\item  (Section~\ref{section:scat}) The so-called ``$\Gamma$-series'' method is outlined, through which it may 
be possible to later calculate  the $l-p$ scattering amplitudes and cross-sections with lowest- and higher-order 
radiative corrections of QED (quantum electrodynamics), by~computing the contributions from one-loop and two-loop 
Feynman diagrams, as~well as from other next-to-leading and higher-order~diagrams.

\item  (Section~\ref{section:con}) At the end, we discuss our results and give a direction for future developments.

\end{itemize}

\section{Preliminaries}\label{section:prelim}

We start by considering Feynman-type scalar integrals given by 
\be
I_{D} = \int_{\zeta_0} \prod_{a} dq_a \prod_{i} \frac{1}{D_i(q)} ,
\label{eq:eqn_scalar}
\ee
where $D_{i}(q_{i} + p_{i}) = (q_{i} + p_{i})^{2} + m_{i}^{2}$ is the $i^{\rm th}$ particle propagator, with 
the momentum $p_{i}$, the mass $m_{i}$, and~
\be
q_{i} = \sum_{a} l_{i,a} q_{a} ,
\label{eq:eqn_qi}
\ee
where $q_{a}$ is the loop momentum and~$l_{i,a}$ is an integer-valued matrix (see notation in~\cite{bogner2010feynman}).

Throughout this paper, we consider dimensionally regulated integrals such that
\be
\prod_{a} dq_{a} = d^{dL} q ,
\label{eq:eqn_ddL}
\ee
where $d$ is the number of dimensions and~$L$ is the number of loops. One should understand $dL$ as $d \times L$. Analytic continuation in the dimensions 
is performed in the usual way: the symbol for the dimension is kept free and the resulting analytic functions, such as Gamma factors, are treated as 
functions of the complex $d$.

It is convenient to compactify the space of loop momenta to the projective space $\mathbb{CP}^{Ld}$ by adding 
the boundary $\mathbb{CP}^{dL-1}$. The~original integration chain $\zeta_0$ is $\mathbb{RP}^{dL}$, which is a 
real projective space class. While the original integration is performed over the Euclidean space $\mathbb{R}^{dL}$, 
the correct homological interpretation in terms of homology classes requires passing to compactification and 
treating the original integration chain as representing a relative homology class. This is the precise homological 
meaning of “integrating from minus infinity to plus infinity”. Analytic continuation is achieved by moving the 
chain $\zeta_{0}$ along with the analytic continuation path. Expressed with some more detail, the~analytic 
continuation is performed by varying the chain representative of the homology class $[\mathbb{PR}^{dL}]$ along 
the path of that continuation. Then, singularities develop through this path such that it is impossible to choose 
a homologous contour differing from the original one by a chain of homologies. This is the basic vanishing cycle 
mechanism, which we shall exploit in this~paper.

\section{Homology Vanishing~Cycles}\label{section:cycle}
\subsection{Vanishing Cycles and Generic~Polynomials}\label{section:local}

This section discusses the choice of contour of integration in the integral of  Formula~(\ref{eq:eqn_scalar}). 
The original integral has $i\epsilon$ prescription. This precisely amounts to the possibility of moving the integration 
chain away from the zero set of the Feynman quadrics. Then, we  have the following:
\be
\zeta \in H_{n}(\mathbb{CP}^{n} - \cup \{D_{i} = 0\}; ~\mathbb{CP}^{n-1} - \cup \{D_{i,\infty} =0\}) .
\label{eq:eqn_zeta}
\ee
In 
 this formula, $D_{i,\infty}$ is the intersection of the zero set of the Feynman quadrics $D_i$ with the hyperplane $\mathbb{CP}^{dL-1}$. In~our dimensional 
regularization scheme, we analytically continue to these $d$s (dimensions), such that the integrand has no pole in the infinitely remote hyperplane. 
This amounts to the usual power-counting~argument.

There is a cycle in even dimensions in the homology group of $H_{n}(\mathbb{CP}^{n})$ (where, for convenience, 
we put $n=dL$ and $\zeta_{0} = \zeta$). This cycle does not change under monodromy, and~ monodromy generates 
cycles that are in the image of the tube map (the tube maps are explained in~\cite{bott1982differential}).
By using the Alexander duality (tube mapping), we can reduce the computation of these cycles to
\be
\zeta^{\prime} \in H_{n-1}(\cup  \{D_{i} =0\}; ~\cup {D_{i,\infty} =0}) ,
\label{eq:eqn_zetaprime}
\ee
where the integration chain $\zeta^{\prime}$ is given in a different homology~group. Note that the Alexander duality and its use in topology of projective 
varieties is explained ~\cite{libgober1994homotopy,dimca2009topology}. Tube mapping can be defined for singular varieties. In this case, the fiber over 
points in the singular set is the intersection of the normal cone with a sphere of small diameter.

The variety $ \cup \{D_{i} = 0\}$ has multiple irreducible components; therefore, the~use of the Mayer--Vietoris spectral sequence 
is appropriate in this case. The Mayer--Vietoris exact sequence can be applied to the tubular neighborhood of the (in general, singular) 
variety $\cap \{D_i=0\}$. It reduces the entire computation to
\be
Z_{I} = H_{n-|I|-1}(\cap \{D_{i}=0\}; ~\cap \{D_{i,\infty}=0\}) ,
\label{eq:eqn_ZI}
\ee
where $I$ is the propagator tuples, $I=(i_{1},...,i_{s})$, a~multi-index in the subscript of the homology group 
of $H_{|I|}$. Formula (\ref{eq:eqn_ZI}) gives a domain, where the vanishing cycles lie (in momentum space representation).

The above Mayer--Vietoris spectral sequence of $Z_{I}$ is applicable to more general integrals, such as
\be
I_{PQ} = \int_{\zeta_{0}} \prod_{a,b} d^{a}q\,d^{b}\bar{q}\,\prod_{i,j} \frac{1}{P_{i}(q)} \frac{1}{\bar{Q}_{j}(q)} ,
\label{eq:eqn_IPQ}
\ee
where we integrate the rational form of the type $a,b$ and $P(q)$/$Q(q)$ are different polynomials but of a generic type. 
$\bar{Q}(q)$ denotes the complex conjugation of $Q(q)$. For~definiteness, we  focus on the integral of the type
\be
I_{P} = \int_{\zeta_{0}} \prod_{a,b} d^{a}q\,d^{b}\bar{q}\,\prod_{i} \frac{1}{P_{i}(q)} .
\label{eq:eqn_IP}
\ee
Polynomials of the type $Q(q)$ are used only in generalized integrals, like in (\ref{eq:eqn_IPQ}).
Based on~(\ref{eq:eqn_IP}), the~Mayer--Vietoris sequence allows us to reduce the homology computation to the computation 
of
\be
H_{|a|+|b|-|I|}(\cap \{P_{i}=0\}) .
\label{eq:eqn_H1}
\ee
In this case, there is more information to utilize. We are actually interested in the image of
\be
H^{|a|-|I|,|b|}_{D}(\cap \{P_{i}=0\}) ,
\label{eq:eqn_H2}
\ee
in $H_{|a|+|b|-|I|}(\cap \{P_{i}=0\})$ under the de Rham isomorphism, where $H^{*,*}_{D}(X)$ is the Dolbeault 
cohomology (see 
Refs.~\cite{brown2012mixed,beem2020secondary} for a discussion of the de Rham vs. 
Dolbeault cohomology in the context of Feynman integrals). By~$H^{p,q}$ in the above, we mean the image Dolbeault cohomology in the cohomology with integer coefficients. 
Consequently, we are particularly interested in the part of cohomology of this specific type. Additionally, there is a dual Mayer--Vietoris sequence that can be made to 
obtain the Hodge filtrations~\cite{peters2008mixed,deligne1975real,schmid1973variation}.

\smallskip
\subsection{Non-Local Vanishing~Cycles}\label{section:nonlocal}

In this section, we analyze non-local vanishing cycles in the context of~\cite{massey2016non}. And~these singularities are quite distinct from the ones considered 
in~\cite{bourjaily2023landau}.

There are strata in the parameter space of the polynomials $P(q)$, which correspond to non-local vanishing cycles. 
These strata can be approached by curves lying strictly inside the complement of the discriminant of $P$. Such 
strata are largely beyond discussion in the literature. For~instance, they are not covered in~\cite{arnold2012singularities,seidel2013lagrangian}; nevertheless, there is some discussion in~\cite{le1979fibre}. 
Here, we only remark that such non-local vanishing in principle admits combinatorial classification in terms of the following:
\begin{itemize}
\item[(i)] The number of derivatives, $\partial^{\alpha} P$, that vanish along the stratum for some set of multi-indices,~
$\alpha$; 
\item[(ii)] The geometric locus, where such vanishing happens.
\end{itemize}

If the geometric locus of the vanishing is given, i.e.,~the locus of $q$, in~which the vanishing takes place, e.g.,~
in terms of parameterization, then finding equations for the locus in the parameter space $p_{\omega}$ ($\omega$ 
here and in the following denotes a typical multi-index, $\omega = (\omega_1,...,\omega_d)$) becomes relatively 
straightforward by applying the elimination techniques based on the Grobner bases. If~only the type of singularity 
is known, then it becomes a modulus problem, and~general GIT (geometric invariant theory) techniques are necessary~\cite{mumford1994geometric}. By~type, we mean scheme-theoretic data that include the dimensionality of the components, 
their own singularities (in terms of the degeneration of the tangent bundle), and~the intersection theory of~components.

GIT techniques are very useful in the analysis of bundles that we consider because these bundles typically possess 
a large degree of symmetry, in~particular toric symmetry. The~precise characterization of these bundles is only possible 
after the techniques of GIT are extended sufficiently, such that they can be used in treating Koszul complexes resulting 
from integration by parts~\cite{GZK}, as~well as  in the conversion from D-modules to corresponding flat connections. 
This subject is far reaching, since it, e.g.,~includes a discussion of bases analogous to the polylogarithm bases for 
the arbitrary divisor's case and for the bundles with logarithmic singularities for its complement. Some discussion of 
these topics is available in~\cite{dimca2009topology}. It is intimately tied with various categories of modules  
associated with the components of the~divisor.

\section{Introduction of Pinch~Map}\label{section:pinch}
\subsection{Pinch Map I: General~Polynomials}\label{section:pinch1}

For the generic stratum of the discriminant locus, the~vanishing occurs at some point. In~fact, for~the generic stratum 
of the GZK discriminant, the~pinching is quadratic. For~higher-order isolated singularities of the type treated in 
Milnor's classic~\cite{milnor2016singular}, the~vanishing is also local, though~it is, in general, non-quadratic~\cite{arnold2012singularities,varchenko1986local}.

\begin{theorem}\label{thm:Th1}
For the generic stratum of the GZK discriminant for a polynomial, $P(q)$, the~vanishing cycles become contractibe to the point
\be
q_{i} = Q_{i}(p_{\omega}) ,
\label{eq:eqn_qi}
\ee
where $Q_{i}$ are rational functions defined everywhere on the main stratum of the GZK discriminant ($Q_{i}(p_{\omega})$ 
must not be confused with $Q(q)$, where the latter is just a local symbol for a generic polynomial, while the former is 
actually the solution for a pinch point). If~the vanishing is still local (i.e., the~solution of the degeneracy equations 
consists of isolated points), then the vanishing locus is given by the restriction of the above map to a substratum of the 
discriminant. The~equation for this stratum is obtained by the elimination of $q_{i}$ from
\be
\partial^{\alpha} P = 0 ,
\label{eq:eqn_P}
\ee
for a given set of $\alpha$.
\end{theorem}
\begin{proof}
The proof can be obtained by the adaptation of the standard Sylvester techniques~\cite{GZK}. In~this regard, the~question 
of practical importance is whether an explicit expression for $Q_{i}$ in terms of the GZK discriminants can be obtained, which is 
a much more challenging task~\cite{cattani2013mixed,dolotin2006introduction}. It involves notions that generalize 
discriminants of exact sequences (the so-called Whitehead torsion~\cite{milnor1966whitehead,anokhina2009resultant}) 
to the case of spectral sequences. In~fact, the~toric nature of the problem can be used for imposing further 
combinatorial symmetry on the problem. It includes generalizing the notion of the spectral sequence for the case of 
multiple gradings (giving rise to multidimensional sheets between  which the differentials act) and for the needs 
of combinatorial elimination theory, going beyond the scope of this paper.
\end{proof}

\smallskip
\subsection{Pinch Map II: Feynman Loop~Integrals}\label{section:pinch2}

Feynman loop integrals present a highly degenerate case for vanishing cycles' analysis. Indeed, each of the integral quadrics has a singularity locus (where 
the derivative vanishes identically) of the dimension $(L-1)d$. Nonetheless, there is a remnant of the locality of pinching, as~we demonstrate in this section. The~
key observation stems from the use of the Mayer--Vietoris spectral sequence. We observe that for each vanishing cycle, in~every 
$H_{n}(\mathbb{CP}^{n} - \cup \{ {D_{i} = 0} \})$, there is a group of propagators $D_{i_1},D_{i_2},...,D_{i_s}$, such that there is a vanishing cycle in 
$Z_{I} = H_{n-|I|}(\cap \{D_{i}\})$. For~simplicity, we focus here on the vanishing cycles at the finite distance of the momentum space, $q_{i}$, rather than 
of the boundary at infinity. This latter case can be analyzed~analogously.

In the following theorem, we analyze the case when the vanishing cycle occurs in the loop momenta $q_{i_1},...,q_{i_k}$. These loop momenta are fixed by the pinch condition. 
The rest of the loop momenta are not  fixed and integration is performed over these~coordinates.

\begin{theorem}\label{thm:Th2}
Vanishing cycles at a finite distance are completely classified by the tuples, $I=(i_1,...,i_s)$, of the propagators, such that there is degeneracy among the tangent 
cones to the varieties $\{D_{i}=0\}$. Such vanishing cycles are contractible to a (pinch) point:

~\vspace{6pt}
\be
q_{i_{k},\mu} = Q_{i_{k},\mu}(p) ,
\label{eq:eqn_qik}
\ee
which is a rational function of the masses and external momenta of interacting (scattering) physical particles. The~rest of the internal (virtual) loop momenta is arbitrary.
\end{theorem}
\begin{proof}
The proof follows from the fact that the degeneracy among tangents to $D_{i}$ can be written as
\be
\sum\limits_{i} a_{i} \frac{\partial D_{i}}{\partial q_{j}} = 0 .
\label{eq:eqn_ai}
\ee
This set of equations with the factors $a_{i}$, where each $a_{i}$ defines a point in a projective space of the corresponding dimension (not all $a_{i}$s are zero),
can be solved with
\be
q_{i} = \sum_a \alpha_{i,a}\,p_{a} ,
\label{eq:eqn_qi2}
\ee
where there can be linear relations among $\alpha_{i,a}$. These coefficients are removed from the system of $D_{i} = 0$ due to the homogeneity requirement.
\end{proof}

It is convenient to introduce a somewhat more general family of integrals as follows:
\be 
F(p) = \int d^nq \prod\limits_{i} \frac{1}{P_i} .
\ee 
We  analyze asymptotic expansion near the Morse-type pinch point at $q=0$. We assume that there is a vanishing cycle in the system, $P_1,...,P_k$, and 
the rest of the polynomials do not vanish at 0. Then, we can write 
\be \label{eq:pl}
P_i = c_i\epsilon+l_i q_\parallel+D_i(q_\parallel,q_\perp),~~i = 1,...,k ,
\ee 
where $\epsilon$ is a small parameter that can be thought of as the value of the corresponding Landau polynomial. $\epsilon = 0$ 
corresponds to the vanishing cycle contracting to the point. We split $q=(q_\parallel,q_\perp),q_\parallel \in \mathbb{C}^m$, such that 
\be 
l_iq = l_iq_\parallel,~~i = 1,...,k ,
\ee 
and $l_i$ spans the space dual to $\mathbb{C}^m$ for some integer, $m$. $D_{i}(q_\parallel,q_\perp)$ is at least quadratic in~$q_\parallel,q_\perp$.

We study the asymptotic expansion of the  integral
\be 
F(p_{i,\omega}) = \int d^dq \prod\limits_{i}{\frac{1}{P_i(q)}} ,
\ee
near the stratum in the parameter space, where there is a Morse singularity in the intersection $\cap \{ q:P_i=0\}$, which we 
assume to be a complete intersection at generic values of~parameters.

\begin{theorem}\label{thm:Th3}
There is the following asymptotic expansion near the regular singularity $L_{I}$: 
\bea
& & \!\!\!\!\!\!\!\!\!\!\!\!\!\!\!\!\!\!\!\!\!\!\!\!\!
F(p_{i,\omega}) = (2\pi i)^k\times
\nonumber \\
& & \,\,
\times \int d^{d-k}q \prod_{j\notin I} \frac{1}{P_j(q_\parallel= 0,q_\perp)}  \,\frac{1}{E + S(q_\perp) + o(q_\perp^2,\epsilon^2,q_\perp \epsilon)} + R ,
\label{eq:eqn_J}
\eea
where $R$ is a regular function near generic points of the discriminantal locus $L_{I} = 0$ (it can have singularities along substrata), and~

~\vspace{6pt}
\begin{gather}
S(q_\perp) =
\begin{vmatrix}
Q_0(q_\perp) & l_{0,1} & ... & l_{0,k}  \\ \nonumber
... \\ \nonumber
Q_k(q_\perp) & l_{k,1} & ... & l_{k,k} \\ \nonumber
\end{vmatrix} ,
\end{gather}
and
\begin{gather}
E=
\begin{vmatrix}
c_0 \epsilon & l_{0,1} & ... & l_{0,k} \\ \nonumber
... \\ \nonumber
c_k \epsilon & l_{k,1} & ... & l_{k,k} \\ \nonumber
\end{vmatrix} .
\end{gather}
The numbers $l_{i,k}$ are the coefficients of the linear terms in \ref{eq:pl}.
The momenta forming the set $I$ in the polynomials $P_j(q_\parallel=0,q_\perp)$ must be evaluated at the 
pinch point $0$. 
\end{theorem}

In this theorem, $k$ polynomials, $P_i$, define the pinch point. The~asymptotic expansion occurs in the last polynomial, $P_0$, evaluated at the pinch point. 
In symmetric notation, we express this using a~determinant.

\begin{proof}
We will prove a slightly more general result. Let us consider the integral
\be
I = \int \,d^{k}x\,d^{n}y\  f(x,y) \prod_{i=0}^{k} \frac{1}{P_{i}(x,y)} ,  
\label{eq:eqn_Ifxy}
\ee
where $f(x,y)$ is some regular function and
\be
P_{i}(x,y) = c_{i}\epsilon + L_{i}(x) + Q_{i}(x,y) + O(x^{2},xy,y^{2}) ,
\label{eq:eqn_Pxy}
\ee
for which we are interested in an asymptotic expansion near $\epsilon = 0$, where $L_{i}$ and $Q_{i}$ are one- and two-variable polynomial functions. By~using a local change in variables,
\be
z_{i} = L_{i}(x) 
\label{eq:eqn_zi} ,
\ee
one can transform the integral in (\ref{eq:eqn_Ifxy}) to the form of
\be
I = \frac{1}{det(L)} \int f(z,y)\,d^{k}z\,d^{n}y \frac{1}{P_0}\prod_{i=1}^{k} \frac{1}{c_{i}\epsilon + z_{i} + Q_{i}(z,y) + O(z^{2},zy,y^{2})} ,
\label{eq:eqn_IdetL}
\ee
where the symbol $P_{0}$ denotes the polynomial with the lowest index. We do not preserve the symmetry between polynomials and choose one of them. 
In the end, the~expressions should be symmetric with respect to permutation of polynomials, and~this is an additional check of our calculation.
From here, we see that the contour can be deformed into an elementary loop around $z_{i} = -(c_{i}\epsilon + Q_{i}(0,y))$~(the original 
contour is the long contour corresponding to the vanishing cycle).
\end{proof}

\hskip -0.6truecm
The following lemma is standard.
\begin{lemma}\label{thm:Lem1}
We have an asymptotic expansion:

\be
K(\xi,\eta) = \int dz\,\frac{1}{z+\xi}\,\frac{1}{z+\eta} \approx 2\pi i\,\frac{1}{\xi - \eta} ,
\label{eq:eqn_Kxi}
\ee
where $\xi$ and $\eta$ are of the same order of magnitude. Here, it is assumed that the contour passes between $z=-\xi$ and $z=-\eta$.
\end{lemma}

\begin{lemma}\label{thm:Lem2}
We have an asymptotic expansion for a generic quadratic form $Q(y)$:
\bea
&& \!\!\!\!\!\!\!\!\!\!\!\!\!\!\!\!\!\!\!\!
K(\epsilon) = \int d^{n}y\,\frac{1}{\epsilon + Q(y) + O(y^{3})} \approx
\nonumber \\
&& \approx \frac{V_{n-1}(1) (\epsilon)^{-1 + n/2} }{\sqrt{det(Q)}}\,\bigl( 1 + R_{1}(\epsilon) \bigr) + R_{2}(\epsilon) ,
\label{eq:eqn_Kdn}
\eea
in which $R_{1,2}$ are regular functions of $\epsilon$ and $R_1(\epsilon)\sim\epsilon$ at $\epsilon = 0$, ~$V_{n-1}(1)$ is the volume of 
the unit sphere of the dimension $n-1$, and $Q$ is a nondegenerate quadratic form.
\end{lemma}
Then, Theorem~\ref{thm:Th3} follows from the above two lemmas. We can also simplify the statement of the theorem to the following form:
\be
J_{\zeta} = A_{I,\zeta}(P_{\parallel})\,L_{I}^{-1 + (d - 1) |I|/2} +  R(p) ,
\label{eq:eqn_Jzeta}
\ee
where $A_{I,\zeta}$ is a function of $P_{\parallel}$, ~$P_{\parallel}$ is the coordinates on the manifold $\{(p,m): L_{I}(p,m) = 0\}$, and~$\zeta$ is a contour of integration.
If $|I|=L$ (the number of loops), then the coefficients $A_{I}$ are algebraic functions. In~general, they involve integration 
over other loop momenta. In~the next section, we write an example of such an $A_{I}$ function involving integration over other 
loop~momenta.

\begin{remark} {The} 
$o()$ term in (\ref{eq:eqn_J}) contains terms that are of a higher order in $\epsilon$ and $q_\perp^2$. 
The branching behavior of the integral comes from the region of momenta, where $q_\perp^2 \sim \epsilon$.
The proof of the theorem essentially states that if a vanishing cycle is formed by transversely intersecting $k$
hypersurfaces, then it is possible to deform the contour of integration in a way that one can take the residue in 
$k-1$ variables, and~for the rest of the variables, $q_{\perp}$, we  have a Morse singularity. This is very far from the physical situation 
when several masses are zero. The~latter case corresponds to degeneracies in the intersections and non-transverse intersections. 
In these cases, we  have a non-Morse~singularity.
\end{remark}

\begin{theorem}\label{thm:Th4}
With the conditions of the previous theorem, we have the following expansion near the generic point of the variety defined by the vanishing of the Landau polynomial $(\{L_{I}=0\})$:
\begin{gather}
I_g(p) = S_g(p_\parallel)L_{I}^{\alpha} + T_g(p_\parallel,L_{I}) ,
\end{gather}
where $\alpha$ is a constant and $S_g(p_\parallel)$ is defined in $\{L_{I}=0\}$, and~we have a regular function at the generic point.
 $T_g(p_\parallel,L_{I})$ is defined 
in a neighborhood of $\{L_{I}=0\}$ and~is regular near the same generic point. This expansion holds for each of the $g$ components 
of the bundle~\cite{Srednyak:2017ime} determined by a Feynman diagram. The~functions $S_{g}$ and $T_{g}$ can develop singularities of 
a regular type near the intersection of $\{L_{I}=0\}$ and other Landau varieties.  
\end{theorem}
\begin{proof}
The proof of this corollary essentially follows from the expressions in the previous theorem and lemmas. For the~leading order, we have 
\be 
S_g = (2\pi i)^{|I|} \prod_{j \notin I} \frac{1}{D_j(Q_\parallel)} .
\label{eq:zalupa}
\ee
We see that for the leading asymptotic term, one can derive a concrete, closed-form expression, while for the regular part (as for $T_g$),
it is impossible to derive a simple expression and only an integral representation may be given. We can also normalize the Landau 
polynomial to a rational expression, and~the final form is given by
\begin{gather} 
L = 
\begin{vmatrix}
P_0(Q_\parallel) & \partial_\parallel P_0(Q_\parallel)  \\ \nonumber
... \\ \nonumber
P_k(Q_\parallel) & \partial_\parallel P_k(Q_\parallel) \nonumber \\
\end{vmatrix} ,
\end{gather}
which is a $(k+1) \times (k+1)$ determinant. $Q_\parallel$ is the restriction of $Q$ in the space orthogonal to the space generated 
by the normals to $P_i$. If~a different normalization was selected, there would be a corresponding factor in  Formula 
(\ref{eq:zalupa}).
\end{proof}

\section{Examples with Feynman Vertex Loop~Diagrams}\label{section:exam}
 
In this section, we consider examples with some Feynman vertex loop diagrams for arbitrary complex masses,
which are discussed at one- and two-loop~levels.

\subsection{Vertex Diagram at One Loop: General~Case}\label{section:vert1}

Vertex diagrams at one loop, such as the example shown in Figure~\ref{fig:1lvertex}, give rise to the 
following type of master integral:
\be 
J_{\rm 1l.v} = \int d^{d}q\,\frac{1}{\prod\limits_{i}((q + p_{i})^{2} + m_{i}^{2})} ,
\label{eq:eqn_J1a}
\ee 
where the product is over at most $d+1$ propagators because we can reduce the general situation (where there is more than $d+1$) to this case by partial fractioning. {The case 
of the QED vertex one-loop diagrams, 
	shown in Figure~\ref{fig:vertex1}, corresponds to a particular choice of masses  with $(0,m_{l},m_{l})$, where $m_{l}$ is the lepton 
	mass. It matches up to a singular stratum in the diagram with generic masses, e.g.,~in the generalized diagram of Figure~\ref{fig:1lvertex}, 
	there is a singularity at $m_{3}=m_{2}=m_{l}$. In~the corresponding QED diagram (the right panel of Figure~\ref{fig:vertex1} in 
	this case), we first restrict to the stratum $m_{3}=m_{2}=m_{l}$ and~then perform analytic continuation}.
\begin{figure}[h!]
\begin{center}
\includegraphics[width=5.cm]{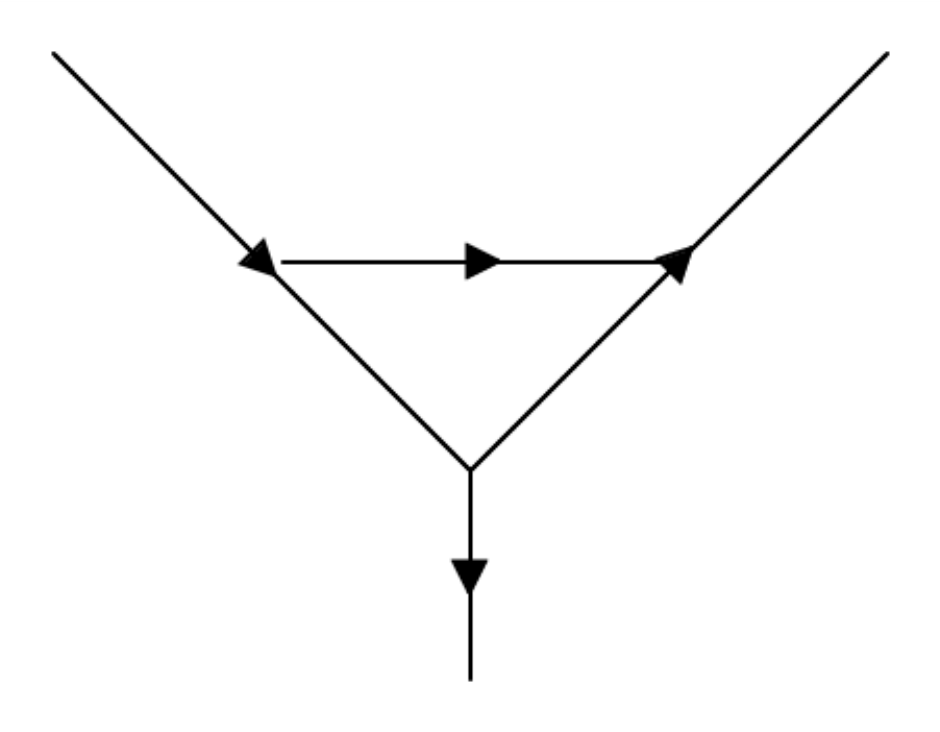}
\caption{The upper 
 vertex of the diagram depicted in the right panel of Figure~\ref{fig:vertex1} but with 
arbitrary complex masses, $m_{i}$, considered. The arrows indicate the flow of momenta.}
\label{fig:1lvertex}
\end{center}
\end{figure}

\begin{figure}[h!]
\includegraphics[width=0.85\textwidth]{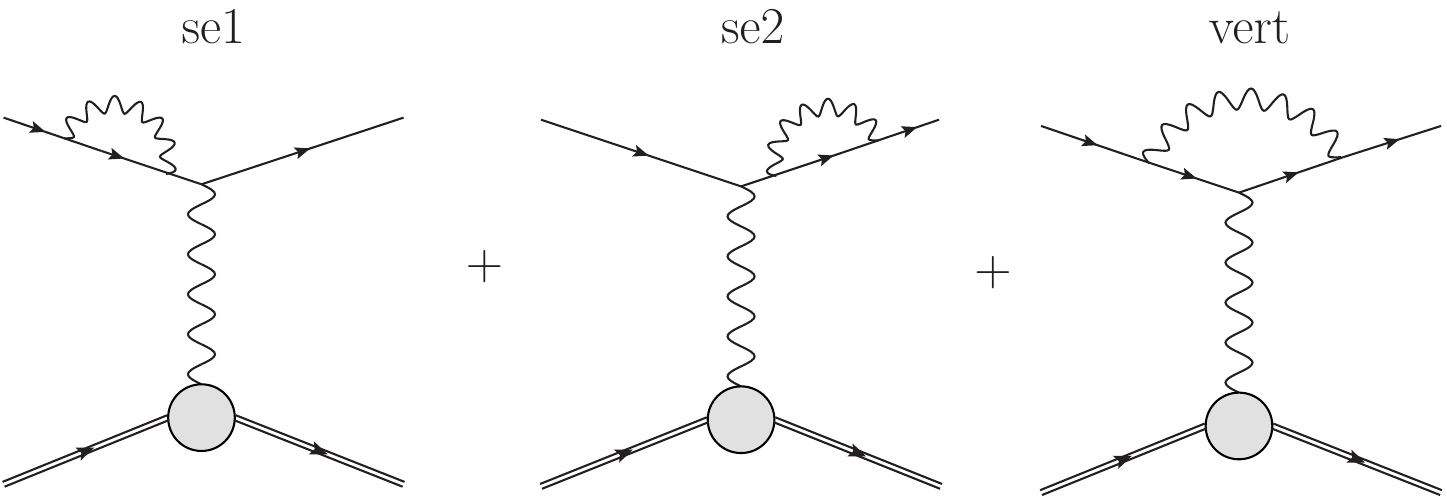}
  \caption{Feynman 
 diagrams describing the QED one-loop corrections at the lepton line in $l-p$ scattering. The arrows indicate the flow of momenta.
  This figure is from Ref.~\cite{Bucoveanu:2018soy} and~is used with the kind permission of \emph{The European Physical Journal A}.}
  \label{fig:vertex1}
\end{figure}%
In this case, the~vanishing cycles arise in any number of propagators from $2$ to $d$. For~definiteness, let us consider 
the propagators $D_{0},...,D_{k}$ given by
\begin{equation}
\begin{gathered}
q^{2} + m_{0}^{2} = 0 , \\ 
(q + p_{i})^{2} + m_{i}^{2} = 0 ,~~i=1 ,..., k .
\end{gathered}
\end{equation}
The variety corresponding to this ideal is equivalent to the one generated by 
\begin{equation}
\begin{gathered}
  q^{2} + m_{0}^{2} = 0 , \\ 
  2q p_{i} + p_{i}^{2} + m_{i}^{2} - m_{0}^{2} = 0 ,~~i=1,...,k ,
\end{gathered}
\end{equation}
or 
\begin{equation}
\begin{gathered}
  q_{\parallel}^{2} + q_{\perp}^{2} + m_{0}^{2} =0 , \\
  q_{\parallel} = \sum\limits_{i} a_{i} p_{i} , \\
  2\sum\limits_{j} (p_{i}p_{j})a_{j} + p_{i}^{2} + m_{i}^{2} - m_{0}^{2} = 0 ,
\label{eq:eqn_Gram}
\end{gathered}
\end{equation}
where $q_{\parallel}$ denotes the component of $q$ that lies in the space spanned by $p_{i}$,
and $q_{\perp}$ is the component of $q$ lying in a subspace orthogonal to the space spanned by $p_{i}$. 

If the Gram determinant of the system (\ref{eq:eqn_Gram}) is $G(p_{i},p_{j}) = 0$, this corresponds to 
a vanishing cycle at infinity (the stratum $\mathbb{CP}^{d-1}$ in the compactification of the $q$-space). 
We consider the case when $G(p_{i},p_{j})\neq 0$. Then, $q_\parallel$ is uniquely determined, and~the 
equation for the singularity is the following:
\begin{gather}
  \begin{vmatrix}
    p_{i}p_{j} & -(p_{i}^{2} + m_{i}^{2} - m_{0}^{2})/2 \\
    -(p_{j}^{2} + m_{j}^{2} - m_{0}^{2})/2 & m_{0}^{2} 
  \end{vmatrix}
= 0 .
\end{gather}
where the determinant has a size of $(d+1) \times (d+1)$. We can now rewrite the integral in~(\ref{eq:eqn_J1a}) as follows:
\be 
J_{\rm 1l.v} = \int d^{d-k} q_{\perp} \int d^{k}q_{\parallel}\,f(q)\,\frac{1}{D_0} 
\prod\limits_{i} \frac{1}{D_{0} + 2p_{i}q_{\parallel} + p_{i}^{2} + m_{i}^{2} - m_{0}^{2}} ,
\label{eq:eqn_J1b}
\ee 
where
\be 
f(q) \equiv f(q_\perp,q_\parallel) = \prod_{j \notin I} \frac{1}{(q + p_{j})^{2} + m_{j}^{2}} ,
\ee 
\be 
D_{0} = q_{\perp}^{2} + q_{\parallel}^{2} + m_{0}^{2} , 
\label{eq:D0}
\ee 
and $f(q)$ is a smooth function equal to the product of the propagators, which are not part of the pinch ($I$ is a 
 set of propagators participating in the pinch).
We then introduce \mbox{the notation} 
\be 
Q_{\parallel} = \sum\limits_{i} A_{i} p_{i} , 
\ee 
where $A_i$ are the solutions of the following equation (the third one of the system (\ref{eq:eqn_Gram})): 
\be 
2\sum\limits_{j} a_{j} (p_{i} p_{j}) + p_{i}^{2} + m_{i}^{2} - m_{0}^{2} = 0\ .
\ee 
This is a point to which the vanishing cycle becomes contractible (the pinch point in our terminology).
We also need to perform the following change in variables:
\be 
q_{\parallel} = Q_{\parallel} + s . 
\ee 
In all these variables, the~integral in (\ref{eq:eqn_J1b}) transforms into
\bea
& & \!\!\!\!\!\!\!\!\!\!\!\!\!\!\!\!\!\!\!\!\!\!\!\!\!
J_{\rm 1l.v} = \int d^{d-k}q_{\perp} \times
\nonumber \\
& & \times \int d^{k}s \prod_{i \in I} \frac{1}{ q_{\perp}^{2} + 
(Q_{\parallel} + p_{i})^{2} + m_{i}^{2} + 2(Q_{\parallel} + p_{i})s + s^{2}} \prod_{j \notin I} \frac{1}{D_{j}} ,
\label{eq:eqn_J1c}
\eea
which can be reduced to 
\be 
J_{\rm 1l.v} = (2\pi i)^{k} \int d^{d-k}q_{\perp} \prod_{j \notin I } \frac{1}{D_{j}(Q_{\parallel},q_{\perp})}
\frac{1}{M_{I}} + R(L_{I}) ,
\label{eq:eqn_J1d}
\ee 
where $R(L_{I})$ is a regular function near the Landau variety $L_{I}$. In~(\ref{eq:eqn_J1d}), we have 
\begin{gather}
M_{I} =
  \begin{vmatrix}
2Q_0 & B_{0} \\ 
.... \\
2(Q_{k} + p_{k}) & B_{k} 
  \end{vmatrix}
  ,
\end{gather}
with
\be 
B_{i} = q_{\perp}^{2} + (Q_{\parallel} + p_{i})^{2} + m_{i}^{2} .
\ee
Note that 
\be 
M_{I} = 2^{k} [p_{1},...,p_{k}]q_\perp^{2} + \psi ,
\label{eq:eqn_MI}
\ee 
where $[p_1,...,p_k]$ is the $k$-volume spanned by the vectors $p_1,...,p_k$,
and 
\begin{gather}
\psi =
  \begin{vmatrix}
2Q_0 & Q_\parallel^{2} + m_{0}^{2} \\ 
.... \\
2(Q_{k} + p_{k}) & (Q_\parallel + p_{k})^{2} + m_{k}^{2}
  \end{vmatrix}
.
\end{gather}
One can 
also indicate that the first $(1, ..., k)$ are coordinates on the space spanned by $p_{i},~i\in [1,...,k]$. The quantity $\psi$ vanishes identically on $L_{I}$. In~fact, it is proportional to the polynomial $L_{I}$.
At the end, we wish to emphasize that the integral in (\ref{eq:eqn_J1d}) has the following \mbox{asymptotic~expansion:}
{
\makeatletter\def\f@size{10.0}\check@mathfonts
\def\maketag@@@#1{\hbox{\m@th\normalsize\normalfont#1}}%
\bea
& & 
\!\!\!\!\!\!\!\!\!\!\!\!\!\!\!\!\!\!\!\!\!\!\!\!\!\!
J_{\rm 1l.v,asym} = 
\nonumber \\
& & \!\!\!\!\!\!\!\!\!\!\!\!\!\!\!\!\!\!\!\!\!\!\!\!\!
= (2\pi i)^{k}\,V_{k-1}(1)(2^{k} [p_{1},...,p_{k}])^{(d-k)/2}\,\psi^{-1 + (d - k)/2} \biggl( \prod_{j \notin I} \frac{1}{D_{j}(Q_{\parallel},0)} 
+ O(\psi) \biggr) + R(L_{I}) ,
\label{eq:eqn_J1e}
\eea
}
where $V_{k-1}(1)$ is the volume dependent on $k-1$ dimensionality. $D_j$ is just the propagator evaluated at $Q_{\parallel}$.


\smallskip
\subsection{Propagator at Two Loops: General~Case}\label{section:propag}

Herein, we begin with the integral of the two-loop propagator, the~diagram of which is shown in Figure~\ref{fig:loop}:
\bea
& & 
\!\!\!\!\!\!\!\!\!\!\!\!\!\!\!\!\!\!\!\!\!
J_{\rm 2l.p} = \int d^dq_{1} d^dq_{2} \times
\nonumber \\
& & \!\!
\times \frac{1}{q_{1}^{2} + m_{1}^{2}}\frac{1}{(q_{1} + p)^{2} + m_{2}^{2}}
\frac{1}{(q_{1} + q_{2})^{2} + m_{3}^{2}}\frac{1}{q_{2}^{2} + m_{4}^{2}}\frac{1}{(q_{2} - p)^{2} + m_{5}^{2}} .
\label{eq:eqn_Jtwo}
\eea
\begin{figure}[h!]
\begin{center}
\includegraphics[width=5.0cm]{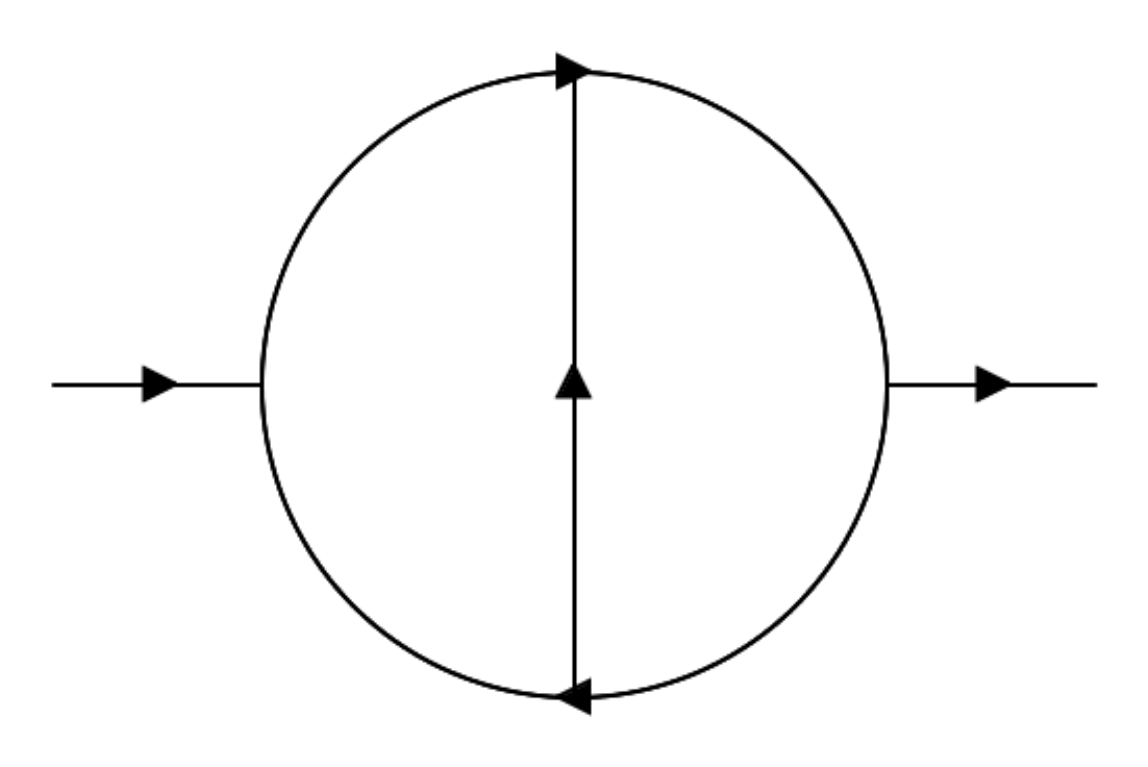}
\caption{A diagram describing the two-loop propagator but with arbitrary complex masses, $m_{i}$, considered.}
\label{fig:loop}
\end{center}
\end{figure}

\hskip -0.6truecm
We will focus on the vanishing cycle in
\begin{equation}
\begin{gathered}
q_{1}^{2} + m_{1}^{2} = 0 , \\
(q_{1} + p)^{2} + m_{2}^{2} = 0 .
\label{eq:eqn_vanish1}
\end{gathered}
\end{equation}
Here, the pinch point is
\be
q^{\prime}_{1} = \alpha p ,
\label{eq:eqn_pinch2}
\ee
which is a solution of the system (\ref{eq:eqn_vanish1}) for some scalar $\alpha$ and~where the singularity locus is at
\be
\sqrt{-p^2} = \pm m_{1} \pm m_{2} ,
\label{eq:eqn_loc2}
\ee
\be 
\alpha = -\frac{p^{2} + m_{2}^{2} - m_{1}^2}{2p^2} .
\ee 
The integral coefficient in the series expansion (explained in Theorem~\ref{thm:Th3}) is expressed~as 
\be
I_{\rm 2l.p} = \int d^{d}q\,\frac{1}{(Q_{1} + q_{2})^{2} + m_{3}^{2}}\frac{1}{q_{2}^{2} + m_{4}^2}\frac{1}{(q_2-p)^{2} + m_{5}^{2}} ,
\label{eq:eqn_coef2}
\ee
and this integral coefficient is a function of masses and dimension only, with~$Q_{1}$ being one of the pinch points.
In this case, the~quadratic form and $\epsilon$-term (see Theorem~\ref{thm:Th3} again) 
are given by
\begingroup
\makeatletter\def\f@size{10.5}\check@mathfonts
\def\maketag@@@#1{\hbox{\m@th\normalsize\normalfont#1}}%
\bea
&& \!\!\!\!\!\!\!\!\!\!\!\!\!\!\!
L^{\prime}_{0} + Q^{\prime}_{0} =
\nonumber \\
&& = 
-2\alpha \sqrt{-p^2} \bigl( (\alpha + 1)^{2}p^{2} + q_{\perp}^{2} + m_{2}^{2} \bigr) +2(\alpha + 1)\sqrt{-p^{2}}
(\alpha^{2}p^{2} + q_{\perp}^{2} + m_{1}^{2}) =
\nonumber \\
&& =
2q_{\perp}^{2} + L^{\prime}_{0} ,
\label{eq:eqn_quad2}
\eea
\endgroup
where $L^{\prime}_{0}$ is the Landau polynomial~\cite{landau1960analytic}, which is obtained by eliminating $\alpha$ from the system~(\ref{eq:eqn_vanish1}):
\be
L^{\prime}_{0} = -2\alpha \sqrt{-p^{2}} \bigl( (\alpha + 1)^{2} p^{2} + m_{2}^{2} \bigr) + 2(\alpha + 1)\sqrt{-p^{2}}(\alpha^{2} p^{2} + m_{1}^{2}) ,
\label{eq:eqn_L02}
\ee
where $\alpha = \alpha(p,m)$ in Equations~(\ref{eq:eqn_quad2}) and~(\ref{eq:eqn_L02}) are certain algebraic expressions deduced as a result of the above-mentioned elimination.
The full expansion has the following~form:
\be
F_{\rm 2l.p} = (2\pi i)^2 I_{\rm 2l.p} \left[L^{\prime}_{0}/2 \right]^{-1 + (d - 1)/2}(1 + O(L^{\prime}_{0})) + R(p) ,
\label{eq:eqn_full2}
\ee

\bigskip
\subsection{Vertex Diagram at Two~Loops}\label{section:vert2}

\subsubsection{General~Case}\label{section:vert2a}

Vertex diagrams at two loops, such as the example shown in Figure~\ref{fig:2lvertex}, 
have the following master integral form: 
\bea
&& \!\!\!\!\!\!\!\!\!\!\!\!\!\!\!\!\!\!\!
J_{\rm 2l.v} = \int d^{d}q_{1} d^{d}q_{2}\,\frac{1}{(q_{1}^{2} + m_{1}^{2})((q_{1} + p_{1})^{2} + m_{2}^{2})(q_{2}^{2} + m_{3}^{2})} \times 
\nonumber \\ 
&& \!\!
\times \frac{1}{((q_{2} + p_{2})^{2} + m_{4}^{2})((q_{1} + q_{2} + p_{1})^{2} + m_{5}^{2})((q_{1} + q_{2} + p_{2})^{2} + m_{6}^{2})} ,
\eea
such that there are six given propagators in the integrand's denominator:
\be
J_{\rm 2l.v} = \int d^{d}q_{1} d^{d}q_{2}\,\frac{1}{D_{1}D_{2}D_{3}D_{4}D_{5}D_{6}} .
\ee

The case of 
QED-vertex two-loop 
	diagrams, shown in Figure~\ref{fig:vertex2}, corresponds to a particular choice of masses with $(0,0,m_{l},m_{l})$. It matches 
	up to a singular stratum in a diagram with generic masses, e.g, in~the generalized diagram of Figure~\ref{fig:2lvertex}, 
	there is a singularity at $m_{4}=m_{3}=m_{2}=m_{1}$. In~the corresponding QED diagram (Figure~\ref{fig:vertex2}b 
	in this case), we first restrict to the stratum $m_{4}=m_{3}=m_{2}=m_{1}$ and~then perform analytic continuation.

\begin{figure}[h!]
\begin{center}
\includegraphics[angle=0,width=6.5cm]{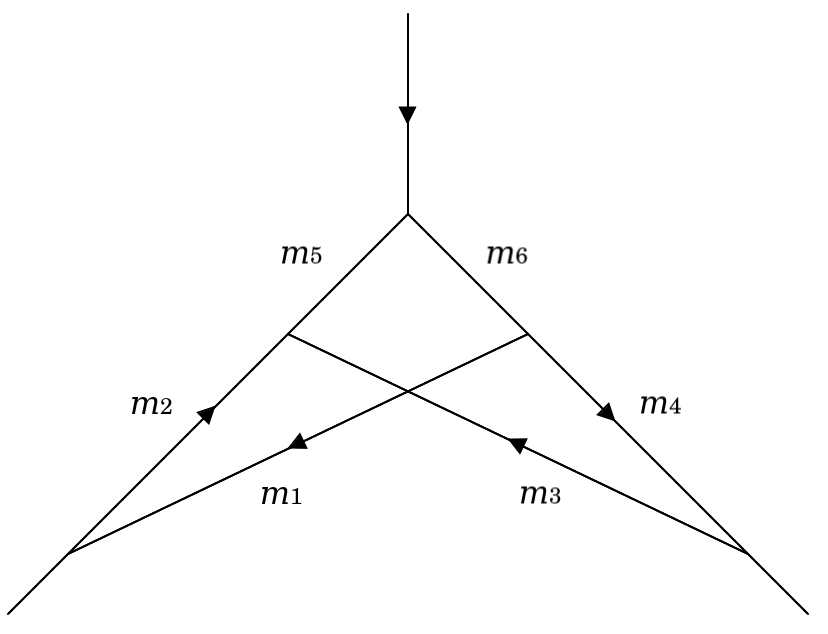}
\caption{The vertex diagram shown in Figure~\ref{fig:vertex2}b but with arbitrary complex masses considered.}
\label{fig:2lvertex}
\end{center}
\end{figure}

\vspace{-15pt}
\begin{figure}[h!]
\vspace{-0.25cm}
\includegraphics[width=0.95\textwidth]{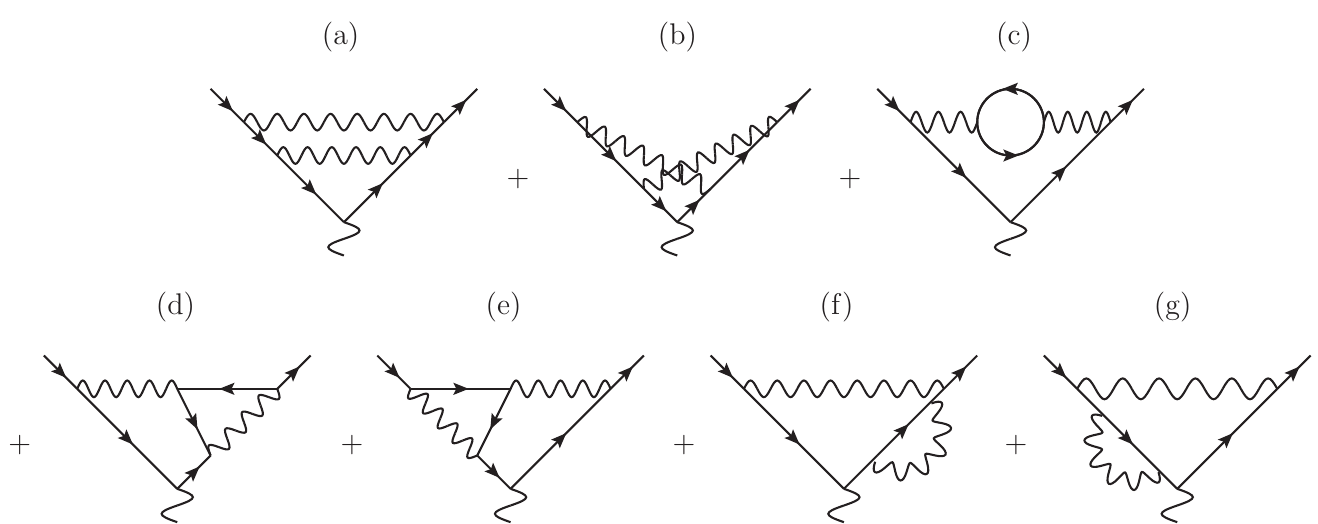}
  \caption{Feynman 
 diagrams describing  QED two-loop corrections at the lepton line in $l-p$ scattering. The diagrams in the subfigures (a)-(g) constitute complete list of two-loop corrections to the QED vertex, 
 and are important in the scattering theory because those constitute the sef of model independent corrections. We refer to Ref.~\cite{Bucoveanu:2018soy} for more details. This figure is from
 Ref.~\cite{Bucoveanu:2018soy} and~is used with the kind permission of \emph{The European Physical Journal A}.}
  \label{fig:vertex2}
\end{figure}%

Let us now concisely discuss the two-pinch singularity of the two-loop topology with the crossed vertex. We  focus on 
the vanishing cycle by considering the two-pinch with the following propagators: 
\begin{equation}
\begin{gathered}
q_{1}^{2} + m_{1}^{2} = 0 , \\ 
(q_{1} + p_{1})^{2} + m_{2}^{2} = 0 .
\label{eq:eqn_vanish_2l_a}
\end{gathered}
\end{equation}
The pinch happens at 
\be 
q_{1}^{\prime\prime} = \alpha p_{1} ,
\label{eq:eqn_pinch_2l_a}
\ee 
for which we make the change in variables of 
\be 
q_{1}^{\prime\prime} = \alpha p_{1} + q_{\parallel} + q_{\perp} .
\ee 
Then, we have 
\be 
J_{\rm 2l.v} = \int d^dq_2 \prod_{j=3}^6
 \frac{1}{D_j} \int dq_\parallel d^{d-1} q_{\perp}\,\frac{1}{D_1 D_2} ,
\ee
or
\bea
& & 
\!\!\!\!\!\!\!\!\!\!\!\!\!\!\!\!\!\!\!\!
J_{\rm 2l.v} = \int d^dq_2 \prod_{j=3}^6 \frac{1}{D_j} \int dq_{\parallel} d^{d-1} q_{\perp}\,
\frac{1}{(\alpha^{2} p_{1}^{2} + m_{1}^{2}) + 2 \alpha p_{1} q_{\parallel} + q_{\parallel}^{2} + q_{\perp}^{2}} \times
\nonumber \\
& & ~~~~~~~~~~~~~~~~~~~
\times \frac{1}{\bigl( (\alpha + 1)^{2} p_{1}^{2} + m_{2}^{2} \bigr) + 2 (\alpha + 1) p_{1} q_{\parallel} 
+ q_{\parallel}^{2} + q_{\perp}^{2}} ,
\eea
which reduces to
{
\makeatletter\def\f@size{10.5}\check@mathfonts
\def\maketag@@@#1{\hbox{\m@th\normalsize\normalfont#1}}%
\be 
J_{\rm 2l.v} = (2\pi i)^2\int d^dq_{2} \prod_{j=3}^6 \frac{1}{D_{j}} \int  d^{d-1} q_{\perp}\,\frac{1}{ 2\sqrt{p_{1}^{2}}} 
\frac{1}{q_{\perp}^{2} + L_{I}} (1 + O(L_{I})) + R(L_{I}) ,
\label{eq:eqn_2loop_final}
\ee
}
where 
\be 
L_{I} = (\alpha + 1)(\alpha^{2}p_{1}^{2} + m_{1}^{2}) - \alpha \bigl( (\alpha + 1)^{2}p_{1}^{2} + m_{2}^{2} \bigr) ,
\label{eq:eqn_Land}
\ee
that is, the corresponding Landau polynomial. The~asymptotic solution corresponding to~(\ref{eq:eqn_2loop_final})
is
{
\makeatletter\def\f@size{10.5}\check@mathfonts
\def\maketag@@@#1{\hbox{\m@th\normalsize\normalfont#1}}%
\be 
J_{\rm 2l.v,asym.} = (2\pi i)^2 \int d^dq_{2} \prod_{j=3}^6 \frac{1}{D_{j}} \frac{1}{ 2\sqrt{p_{1}^{2}}} V_{d-2}(1) L_{I}^{(d-3)/2}(1 + O(L_{I})) + R(L_{I}) ,
\ee
}
which is similar to the formula (\ref{eq:eqn_full2}).

\smallskip
\subsubsection{Five-Pinch Point for Arbitrary Mass~Vertex}\label{section:vert2b}

In an asymptotic analysis of the arbitrary-mass two-loop ladder crossed vertex as shown in Figure~\ref{fig:2lvertex}, 
the geometry of the Landau varieties simplifies, such that it is possible to carry out the analysis of the five-pinch point explicitly. 
In this case, we consider the five-pinch point by the following propagators:
\begin{equation}
\begin{gathered}
q_{1}^{2} + m_{1}^{2} = 0 , \\
(q_{1} + p_{1})^{2} + m_{2}^{2} = 0 , \\
q_{2}^{2} + m_{3}^{2} = 0 , \\
(q_{2} + p_{2})^{2} + m_{4}^{2} = 0 , \\
(q_{1} + q_{2} + p_{1})^{2} + m_{5}^{2} = 0 .
\label{eq:eqn_5pinch_a}
\end{gathered}
\end{equation}
We have the pinch condition of 
\be
q_{i} = \alpha_{i}\,p_{1} + \beta_{i}\,p_{2} .
\label{eq:eqn_5pinch_b}
\ee
For example, (\ref{eq:eqn_5pinch_b}) and the first two formulas of (\ref{eq:eqn_5pinch_a}),
\begin{equation}
\begin{gathered}
q_{1} = \alpha_{1} p_{1} + \beta_{1} p_{2} , \\ 
2q_{1}\,p_{1} + p_{1}^{2} + m_{2}^{2} - m_{1}^{2} = 0 , \\
q_{1}^{2} + m_{1}^{2} = 0 .
\label{eq:eqn_5pinch_c}
\end{gathered}
\end{equation}
allow us to determine $\alpha_{1}$ and $\beta_{1}$. Analogously, (\ref{eq:eqn_5pinch_b}) and the third and fourth formulas of~(\ref{eq:eqn_5pinch_a}), \vspace{6pt}
\begin{equation}
\begin{gathered}
q_{2} = \alpha_{2} p_{1} + \beta_{2} p_{2} , \\ 
2q_{2}\,p_{2} + p_{2}^{2} + m_{4}^{2} - m_{3}^{2} = 0 , \\
q_{2}^{2} + m_{3}^{2} = 0 .
\label{eq:eqn_5pinch_d}
\end{gathered}
\end{equation}
allow to determine $\alpha_{2}$ and $\beta_{2}$.

The asymptotic analysis for the arbitrary mass vertex proceeds as follows. Near~the pinch points $(Q_{1}$ and $Q_{2})$, 
we have the following integral: 
\begingroup
\makeatletter\def\f@size{9.5}\check@mathfonts
\def\maketag@@@#1{\hbox{\m@th\normalsize\normalfont#1}}%
\bea
& & \!\!\!\!\!\!\!\!\!\!
J_{\rm 5pinch,asym.} = \int d^{d-2}q_{1,\perp} d^{d-2}q_{2,\perp} \int d^{2}q_{1,\parallel} d^{2}q_{2,\parallel}\,
\frac{1}{Q_{1}^{2} + m_{1}^{2} + 2Q_{1}\,q_{1,\parallel} + q_{1,\parallel}^{2} + q_{1,\perp}^{2}} 
\nonumber \\
& & \!\!\!\!\!\!\!\!\!\!
\times \frac{1}{(Q_{1} + Q_{2} + p_{1})^{2} + m_5^2 + 2(Q_{1} + Q_{2} + p_{1})(q_{1,\parallel} + q_{2,\parallel}) + (q_{1,\parallel} + 
q_{2,\parallel})^{2} + (q_{1,\perp} + q_{2,\perp})^{2}}
\nonumber \\
& & \!\!\!\!\!\!\!\!\!\!
\times \frac{1}{(Q_{1} + Q_{2} + p_{2})^{2} + m_6^2 + 2(Q_{1} + Q_{2} + p_{2})(q_{1,\parallel} + q_{2,\parallel}) + (q_{1,\parallel} + 
q_{2,\parallel})^{2} + (q_{1,\perp} + q_{2,\perp})^{2}}
\nonumber \\
& & ~~~~~~~~~~~~
\times \frac{1}{(Q_{1} + p_{1})^{2} + m_{2}^{2} + 2(Q_{1} + p_{1})q_{1,\parallel} + q_{1,\parallel}^{2} + q_{1,\perp}^{2}}
\nonumber \\
& & ~~~~~~~~~~~~~~~~~~~~
\times\frac{1}{Q_{2}^{2} + m_{3}^{2} + 2Q_{2}\,q_{2,\parallel} + q_{2,\parallel}^{2} + q_{2,\perp}^{2}}
\nonumber \\
& & ~~~~~~~~~~~~~
\times \frac{1}{(Q_{2} + p_{2})^{2} + m_{4}^{2} + 2(Q_{2} + p_{2})q_{2,\parallel} + q_{2,\parallel}^{2} + q_{2,\perp}^{2}} ,
\label{eq:eqn_5pinch_e}
\eea
\endgroup
which can be rewritten as 
\begingroup
\makeatletter\def\f@size{9.0}\check@mathfonts
\def\maketag@@@#1{\hbox{\m@th\normalsize\normalfont#1}}%
\begin{equation}
\begin{gathered}
J_{\rm 5pinch,asym.} = (2\pi i)^{4} \int d^{d-2}q_{1,\perp} d^{d-2}q_{2,\perp} \frac{1}{D_6}\times \\ 
\begin{vmatrix}
Q_{1} & 0 & Q_{1}^{2} + m_{1}^{2} + q_{1,\parallel}^{2} + q_{1,\perp}^{2}  \\
Q_{1} + p_{1} & 0  & (Q_{1} + p_{1})^{2} + m_{2}^{2} + q_{1,\parallel}^{2} + q_{1,\perp}^{2} \\ 
0 & Q_{2} & Q_{2}^{2} + m_{3}^{2} + q_{2,\parallel}^{2} + q_{2,\perp}^{2} \\
0 & Q_{2} + p_{2} & (Q_{2} + p_{2})^{2} + m_{4}^{2} + q_{2,\parallel}^{2} + q_{2,\perp}^{2} \\
Q_{1} + Q_{2} + p_{1} & Q_{1} + Q_{2} + p_{1} & (Q_{1} + Q_{2} + p_{1})^{2} + m_5^2 + (q_{1,\parallel} + 
q_{2,\parallel})^{2} + (q_{1,\perp} +q_{2,\perp})^{2}
\end{vmatrix}^{-1}
\label{eq:eqn_5pinch_f}
\end{gathered}
\end{equation}
\endgroup
where $q_{1,\parallel}(q_{1,\perp},q_{2,\perp})$ and $q_{2,\parallel}(q_{1,\perp},q_{2,\perp})$ are $q_\perp$-dependent residues. In general, 
the~subscripts $v_{1,2\parallel}$ and $v_{1,2\perp}$ in our formulas mean components of the orthogonal decomposition of the vector 
$v$ with respect to the space spanned by $p_{1,2}$, respectively.

In this expression, a~determinant with vector entries has the following meaning. The~vectors $Q_{1,2}$ are essentially two-dimensional (due to Landau conditions, 
they are always a linear combination of $p_{1,2}$). Then, a column in which the combination of $Q_i$ and $p_i$ occurs is meant to be split in two columns containing 
the coefficients in front of $p_1$ and $p_2$, respectively. The~general formula is in the Theorem~3.

For the leading order in the expansion parameter $L_I$ (the formula defining the Landau polynomial $L_{I}$ in (\ref{eq:eqn_5pinch_h})), ~Equation~(\ref{eq:eqn_5pinch_f}) 
reduces to 
\begingroup
\makeatletter\def\f@size{9.5}\check@mathfonts
\def\maketag@@@#1{\hbox{\m@th\normalsize\normalfont#1}}%
\begin{equation}
\begin{gathered}
J_{\rm 5pinch,asym.} = (2\pi i)^4 \int d^{d-2}q_{1,\perp} d^{d-2}q_{2,\perp} \times \\ 
\times
\begin{vmatrix}
Q_{1} & 0 & Q_{1}^{2} + m_{1}^{2} + q_{1,\perp}^{2}  \\
Q_{1} + p_{1} & 0  & (Q_{1} + p_{1})^{2} + m_{2}^{2} + q_{1,\perp}^{2} \\ 
0 & Q_{2} & Q_{2}^{2} + m_{3}^{2} + q_{2,\perp}^{2} \\
0 & Q_{2} + p_{2} & (Q_{2} + p_{2})^{2} + m_{4}^{2} + q_{2,\perp}^{2} \\
Q_{1} + Q_{2} + p_{1} & Q_{1} + Q_{2} + p_{1} & (Q_{1} + Q_{2} + p_{1})^{2} + m_{5}^{2} + (q_{1,\perp} + q_{2,\perp})^{2}
\end{vmatrix}^{-1} .
\label{eq:eqn_5pinch_g}
\end{gathered}
\end{equation}
\endgroup
Note that in this case we have a determinantal formula for the Landau polynomial up to an overall factor $A$, which is a rational factor:
\begingroup
\makeatletter\def\f@size{9.5}\check@mathfonts
\def\maketag@@@#1{\hbox{\m@th\normalsize\normalfont#1}}%
\begin{gather}
L_{I} = A 
\begin{vmatrix}
Q_{1} & 0 & Q_{1}^{2} + m_{1}^{2}   \\
Q_{1} + p_{1} & 0  & (Q_{1} + p_{1})^{2} + m_{2}^{2} \\ 
0 & Q_{2} & Q_{2}^{2} + m_{3}^{2} \\
0 & Q_{2} + p_{2} & (Q_{2} + p_{2})^{2} + m_{4}^{2} \\
Q_{1} + Q_{2} + p_{1} & Q_{1} + Q_{2} + p_{1} & (Q_{1} + Q_{2} + p_{1})^{2} + m_{5}^{2}
\end{vmatrix} ,
\label{eq:eqn_5pinch_h}
\end{gather}
\endgroup
The Landau polynomial in (\ref{eq:eqn_5pinch_h}) is considerably more complicated than that 
in the case of QED, derived in Section~\ref{section:vert2d}.

The determinant in the integral of (\ref{eq:eqn_5pinch_g}) has the following explicit form:
\bea
& & 
\!\!\!\!\!\!\!\!\!\!\!\!\!\!\!\!\!
\mbox{Det} = \beta_{1} \alpha_{2}\,\, \Delta(p_1,p_2) \times
\nonumber \\
& &
\times \bigl( a_{5} - a_{2} - a_{4} - \sigma_{1} (a_{2} - a_{1}) - \rho_{1} a_{1} - \sigma_{2} (a_{4} - a_{3}) - \rho_{2} a_{3} \bigl) ,
\label{eq:eqn_5pinch_i}
\eea
where $\Delta(p_1,p_2) = (p_1^2p_2^2-(p_1p_2)^2)$, $a_{i}$ are the elements of the last column, and~the pinch points are given by
\begin{equation}
\begin{gathered}
Q_{2} = \sigma_{1} p_{1} + \rho_{1} Q_{1} , \\ 
Q_{1} = \sigma_{2} p_{2} + \rho_{2} Q_{2} .
\label{eq:eqn_5pinch_j}
\end{gathered}
\end{equation}

\smallskip
\subsubsection{QED~Case}\label{section:vert2c}

We proceed by considering the above-mentioned QED two-loop crossed vertex diagram in Figure~\ref{fig:vertex2}a as a specific 
example. This part of our analysis is a classical one, which may be compared to~\cite{sudakov1956vertex,gribov1972deep}.
In this case, the~integral of the vertex diagram is obtained from (\ref{eq:eqn_5pinch_a}) by taking the limit $m_{1} \rightarrow 0$,
$m_{3} \rightarrow 0$, $m_{2} = m_{4} = m_{5} \rightarrow m_{l}$. 
\bea
& &
J_{\rm QED:2l.v} = \int d^dq_{1} d^dq_{2}\,\frac{1}{q_{1}^{2}(q_{1} - q_{2})^{2}} \frac{1}{((p_{1} + q_{1})^{2} + m_{l}^{2})
((p_{2} + q_{1})^{2} + m_{l}^{2})} \times 
\nonumber \\
& & 
~~~~~~~~~~~~~~~~~~
\times \frac{1}{((p_{1} + q_{2})^{2} + m_{l}^{2})((p_{2} + q_{2})^{2} + m_{l}^{2})} .
\label{eq:eqn_JtwoQED}
\eea
Note that the mass limit is obtained by taking the limit on the integrand first. If~we were to take the integral 
and then take the limit, we would observe a singularity in $m_{i}-m_{j}$ on an unphysical~sheet. 

We focus on the pinch point by the two fermion propagators $(p_{1} + q_{1})^{2} + m_{l}^{2}$ and 
$(p_{2} + q_{1})^{2} + m_{l}^{2}$. Using the designation
\be 
q = p_{2} - p_{1} , 
\ee 
the pinch  occurs at 
\be 
e_{l} = \frac{q^{2}}{4} + m_{l}^{2} = 0 ,
\ee 
and at 
\be 
Q_{i} = -p_{1} - \frac{1}{2}\,q \ i = 1,2.
\label{eq:eqn_Landpol1}
\ee 
The subtlety here is that the two pinch points in each of the variables coincide with each other, and~there is 
an extra singularity coming from the photon propagator. We then use the following change in variables:
\be 
q_{i} = Q_{i} + s_{i} + q_{i,\perp} ,
\ee 
where $s_{i}$ lies in the subspace spanned by $q$. After~this change in variables, we obtain 
\bea
& &
\!\!\!\!\!\!\!\!\!\!\!\!\!\!
J_{\rm QED:2l.v} = \int ds_{1}\,dq_{1,\perp}\,ds_{2}\,dq_{2,\perp}\,\frac{1}{s_{1}^{2} + q_{1,\perp}^{2}} 
\frac{1}{(q_{1,\perp} - q_{2,\perp})^{2} + (s_{1} - s_{2})^{2}} \times
\nonumber \\
& &
~~~~~~~~~~~~~~~~~~~~~~
\times \frac{1}{(q_{1,\perp}^{2} - q s_{1} + s_{1}^{2} + e_{l})(q_{1,\perp}^{2} + q s_{1} + s_{1}^{2} + e_{l})} \times
\nonumber \\
& &
~~~~~~~~~~~~~~~~~~~~~~
\times \frac{1}{(q_{2,\perp}^{2} -q s_{2} + s_{2}^{2} + e_{l})(q_{2,\perp}^{2} + q s_{2} + s_{2}^{2} + e_{l})} ,
\eea
where one needs to compute the asymptotics at $e_{l} \rightarrow 0$. The factor 
$(p_1+p_2)^2$ comes from $q_1^2$ evaluated at $(-p_1-q/2) = -(p_1+p_2)/2$. The power counting gives $(2(d-1)-6)/2=d-4$ for $e_l$.
\begin{equation}
\begin{gathered}
J_{\rm QED:2l.v,asym.} \approx \frac{1}{q^{2}}\frac{1}{(p_{1} + p_{2})^{2}}\int \frac{d^{d - 1}q_{1,\perp}
\,d^{d-1}q_{2,\perp}}{(q_{1,\perp} - q_{2,\perp})^{2}} 
\biggl(\frac{1}{e_{l} + q_{1,\perp}^2} \frac{1}{e_{l} + q_{2,\perp}^2}\biggr) ,
\end{gathered}
\label{eq:eqn_QED_as}
\end{equation}
which eventually gives 
\be 
J_{\rm QED:2l.v,asym.} = \frac{1}{q^{2}}\frac{1}{(p_{1} + p_{2})^{2}}\,C(d)\,e_{l}^{d-4}(1 + O(e_{l})) ,
\label{eq:eqn_QED_asym}
\ee 
with $C(d)$ as a $d$-dependent~factor.

\smallskip
\subsubsection{Five-Pinch Point for QED~Vertex}\label{section:vert2d}

In an asymptotic analysis of the two-loop ladder contribution to the QED vertex with two massless lines as shown 
in Figure~\ref{fig:vertex2}b, there is a direct analogy with what is already discussed in Section~\ref{section:vert2b}.
In this case, the~five-pinch point is given by the following~propagators:
\begin{equation}
\begin{gathered}
q_{1}^{2} = 0 , \\
(q_{1} + p_{1})^{2} + m_{l}^{2} = 0 , \\ 
q_{2}^{2} = 0 , \\ 
(q_{2} + p_{2})^{2} + m_{l}^{2} = 0 , \\ 
(q_{1} + q_{2} + p_{1})^{2} + m_{l}^{2} = 0 . 
\label{eq:5pinchQED_a}
\end{gathered}
\end{equation}
The pinch condition is given by the same formula (\ref{eq:eqn_5pinch_b}):
\be
q_{i} = \alpha_{i}\,p_{1} + \beta_{i}\,p_{2} . 
\label{eq:eqn_5pinchQED_b}
\ee
Due to another condition of $q_{i}^{2} = 0$, one can write 
\be 
\beta_{i} = s_{i}\,\alpha_{i} , 
\label{eq:eqn_5pinchQED_c}
\ee 
where 
\be 
s_{1,2} = \frac{-p_{1}\,p_{2} \pm \sqrt{G(p_{1}, p_{2})}}{p_{1,2}^{2}} , 
\label{eq:eqn_5pinchQED_d}
\ee 
where $G(p_{1}, p_{2})$ is the square root of the Gram matrix of $p_{i}$ (the signs $\pm$ could be chosen independently for $s_{1,2}$). 
We can also explicitly solve for $\alpha_{i}$ as follows:
\be 
\alpha_{i} = -\frac{p_{1}^{2} + m_{l}^{2}}{2(p_{1}^{2} + 2s_{1} p_{1} p_{2})} .
\label{eq:eqn_5pinchQED_e}
\ee 
Then, the Landau polynomial for the five-pinch point is derived as follows: 
\be 
\bigl(\alpha_{1}(p_{1} + s_{1} p_{2}) + \alpha_2(p_{1} + s_{2} p_{2}) + p_{1}\bigr)^{2} + m_{l}^{2} = 0 ,
\label{eq:eqn_5pinchQED_f}
\ee 
in which there are polynomials for each choice of the signs of $s_{i}$.

The asymptotic analysis for the QED vertex holds similarly to the analysis for the arbitrary mass vertex discussed in
Section~\ref{section:vert2b}. We have the following integral near the pinch points $(Q_{1}$ and $Q_{2})$:
\begin{equation}
\begin{gathered}
    J_{\rm QED:5pinch,asym.} = (2\pi i)^{4} \int d^{d-2}q_{1,\perp} d^{d-2}q_{2,\perp} \frac{1}{(Q_{1} + Q_{2} + p_{2})^{2} + m^{2}} \times
    \\ 
    \times
\begin{vmatrix}
Q_{1} & 0 & Q_{1}^{2} + q_{1,\perp}^{2}   \\
Q_{1} + p_{1} & 0  & (Q_{1} + p_{1})^{2} + m^{2} + q_{1,\perp}^{2} \\ 
0 & Q_{2} & Q_{2}^{2} + q_{2,\perp}^{2}\\
0 & Q_{2} + p_{2} & (Q_{2} + p_{2})^{2} + m^{2} + q_{2,\perp}^{2} \\
Q_{1} + Q_{2} + p_{1} & Q_{1} + Q_{2} + p_{1} & (Q_{1} + Q_{2} + p_{1})^{2} + m^{2} + (q_{1,\perp}+q_{2,\perp})^{2}
\end{vmatrix}^{-1}
    +
    \\ 
+ O(L_{I}) .
\label{eq:eqn_5pinchQED_g}
\end{gathered}
\end{equation}
The evaluation of this determinant gives 
\bea
& & 
\!\!\!\!\!\!\!\!\!\!\!\!\!\!\!\!\!
\mbox{Det} = \beta_{1} \alpha_{2}\,\,\Delta^{2}(p_{1},p_{2}) \times
\nonumber \\
& &
\times \bigl( a_{5} - a_{2} - a_{4} - \sigma_{1} (a_{2} - a_{1}) - \rho_{1} a_{1} - \sigma_{2} (a_{4} - a_{3}) - \rho_{2} a_{3} \bigl) ,
\label{eq:eqn_5pinchQED_h}
\eea
where $a_{i}$ are the elements of the last column, and~the pinch points are given by
\begin{equation}
\begin{gathered}
Q_{2} = \sigma_{1} p_{1} + \rho_{1} Q_{1} , \\ 
Q_{1} = \sigma_{2} p_{2} + \rho_{2} Q_{2} .
\label{eq:eqn_5pinchQED_i}
\end{gathered}
\end{equation}
Ultimately, we obtain up to a factor:
\begin{equation}
\begin{gathered}
\mbox{Det} = a_{5} - a_{2} - a_{4} - \frac{\beta_{2}}{\beta_{1}} a_{1} - 
\frac{1}{\beta_{1}}(\alpha_{2} \beta_{1} - \alpha_{1} \beta_{2})(a_{2} - a_{1}) - \\
\!\!\!\!\!\!\!\!\!\!\!\!\!\!\!\!\!\!\!\!
- \frac{\alpha_{1}}{\alpha_{2}} a_{3} - \frac{1}{\alpha_{2}} (\alpha_{1} \beta_{2} - \alpha_{2} \beta_{1})(a_{4} - a_{3}) ,
\label{eq:eqn_5pinchQED_j}
\end{gathered}
\end{equation}
\be
\mbox{Det} = \epsilon_1 + Q(q_{1,\perp},q_{2,\perp}) .
\label{eq:eqn_5pinchQED_k}
\ee
In (\ref{eq:eqn_5pinchQED_k}), $\epsilon_1$ is proportional to the corresponding Landau polynomial $L_{I}$:
\be 
\epsilon_1 = BL_{I} ,
\label{eq:eqn_5pinchQED_l}
\ee 
where $B$ is some algebraic~function.

We should note that analogous expressions for the general mass case are more complicated and are much longer, which is the reason 
we do not write them after  Equation~(\ref{eq:eqn_5pinch_j}) in Section~\ref{section:vert2b}.

\section{Outlining the Method of \boldmath$\Gamma$ Series for Performing Potentially New Types of Lepton--Proton Cross-Section
Calculations in the~Future}\label{section:scat}

In this section, we concisely describe our vision for the use of the so-called $\Gamma$-series method, how it ties in with 
the vanishing-cycle analysis discussed in Sections~\ref{section:cycle} and~\ref{section:pinch},
and how to employ it for future computations of the scattering amplitudes and cross-sections with lowest- and higher-order 
radiative corrections in $l-p$ scattering processes. Some diagrams presented in Section~\ref{section:exam} are just examples 
of how we shall be treating them in our framework; nevertheless, for~the cross-section calculations, much work still needs 
to be carried out. Here, we conjecture about such a possible future developments because neither this particular section nor the 
other parts of this paper address cross-sections or any other physical observables. Additionally, our current discussion is 
somewhat~generic.

Notwithstanding that, at~this point, one can outline some features in the $\Gamma$-series method. This method is very well 
adapted to the calculation of hypergeometric (HG) Euler integrals with free polynomials in the integrand. In~this case, it 
is possible to exhibit a basis of solutions of the corresponding holonomic system given by $\Gamma$ series 
via Ore--Sato coefficients~\cite{gel1992general}: 
\be 
HG(x;L) = \sum_{\gamma \in L^{\prime}} \frac{x^{v + \gamma}}{\prod\limits_{\omega} \Gamma(v_{\omega} + \gamma_{\omega} + 1)} ,
\label{eq:eqn_Gamma}
\ee
which is associated with a lattice, $L$. In~our case, it is a lattice dual to the lattice generated by the Minkowski sum of Newton polytopes 
of the  Feynman~quadrics. 

In fact, 
	HG series have long been applied to Feynman integrals (see, e.g.,~\cite{Fleischer:2003rm,Yost:2011wk}), and~especially in recent years, this 
	approach has been revived and systematically developed~\cite{ananthanarayan2023feyngkz,delaCruz:2019skx,Feng:2019bdx,Klausen:2019hrg}.
So, in~ Formula (\ref{eq:eqn_Gamma}), $L^{\prime}$ is a lattice dual to $L$; the term 
$1\big/\bigl( \prod \Gamma(v_{\omega} + \gamma_{\omega} + 1 \bigr)$ is the Ore--Sato coefficients; $v$ is the so-called initial 
monomials; $v_\omega$ is specific to an HG system and~in physical applications is a function of the regulator; and 
$\gamma_\omega$ runs over the lattice $L^{\prime}$. The~fact that such formulas unify many known expressions in particle 
physics can be indicated 
 by the fact that the Ore--Sato coefficients, on one hand, provide the dimensional regularization of  
ultraviolet (UV) and infrared (IR) divergent integrals and, on~the other hand, naturally include the $\varepsilon$-coefficients 
of $\zeta$ functions on number fields in the context of HG motives~\cite{beilinson2001epsilon}. 

In this paper, we focus on the asymptotic region near single Landau varieties or, stated in more technical terms, on~the asymptotic 
analysis of the codimension $1$ near Landau varieties. In~this case, such an asymptotic aspect is rather simple: $a L^{\alpha} + b$ for regular 
$a$ and $b$. Therefore, we do not really use any $\Gamma$ series, which is another factor to consider. At~this point, 
the relationship of our vanishing cycle analysis to  $\Gamma$ series is somewhat complicated because  $\Gamma$ 
series construct the duality 
to an asymptotic analysis near intersections of several components of Landau hypersurfaces. 
We are going to elaborate on this in a future publication. In~particular, for~the case of one-loop functions, the~correspondence 
between Landau varieties and $\Gamma$ series has been taken much further in~\cite{aomoto1977analytic}. In~that case, the~
Landau varieties are strata of minuscule Grassmannians, and~a one-loop function is obtained as a lift of a logarithmic bundle
on the iterated covers that are branched with the degree $2$ on the Landau varieties. Stated generally, the~vanishing cycles 
correspond to the expansion of $\Gamma$ series near codimension-1 strata, while the $\Gamma$ series contain precise branching 
information near deeper~strata.

The above considerations are valid only for the generic stratum, where the coefficients of the Feynman quadrics can be deformed 
away from the physical locus~\cite{srednyak2018universal}. This locus, in turn, is a tiny subspace inside the universal deformation. 
Therefore, for~achieving physical results, it is necessary to make use of the machinery of $b$-functions~\cite{saito1994microlocal}.
From a geometric point of view, we have to consider the restriction of the flat GZK bundle to a deep stratum. The~most naive 
practical way to realize this is to start with the Gauss--Manin connection on the generic stratum and restrict it to the physical 
stratum. One complication is that the explicit form of this connection is not known, and~a substantial development of the GZK 
formalism is needed for accomplishing~it.

Furthermore, still from the geometric point of view, it is crucial to consider the uniformization of the singularity loci and 
the lift of the bundle to such uniformization. On~general grounds, one can expect that such uniformization is achieved 
by $\vartheta$-functions~\cite{mumford2007tata}. This is well known for single-variate polynomials~\cite{mumford2007tata,sturmfels2000solving,passare2005singularities}. However, we need to have its multivariate analogy
at our disposal, which requires significant reformulation of the formalism. Note that the $\vartheta$ series constructed in 
the latter-cited papers are particular cases of $\Gamma$ series.

When applied to concrete-amplitude and cross-section calculations with lowest- and higher-order radiative corrections, e.g.,~
in $l-p$ scattering, then along with carrying out the above analysis, we will also require an intimate study of the geometry 
of  Landau polynomials, such as those shown in Formulas (\ref{eq:eqn_L02}), (\ref{eq:eqn_Land}), (\ref{eq:eqn_5pinch_h}), 
and (\ref{eq:eqn_5pinchQED_f}).

\section{Discussion}\label{section:con}

In this paper, we used the Mayer--Vietoris spectral sequence for analyzing the vanishing cycles of Feynman integrals in 
Sections~\ref{section:cycle} and \ref{section:pinch}, where we also gave a complete classification of the singularity 
development mechanisms. 
In Section~\ref{section:pinch1}, we obtained a general theorem of the localization of vanishing cycles in the integrals of polynomials.
In Section~\ref{section:pinch2}, we obtained an integral representation for the coefficients of such an asymptotic expansion, 
as well as an asymptotic expression for integrals with generic polynomials (a result that seems to be new in the mathematical 
literature). In~the whole of Section~\ref{section:exam}, we presented computations of some example diagrams in the case of arbitrary 
complex masses and QED, based on our findings in the previous sections. The~novelty of this paper could be considered the 
approach used to carry out an asymptotic analysis that employs the momentum representation, where the Feynman parameter representation 
is usually~employed.

As it is mentioned in Section~\ref{section:scat}, along with a future development of the entire ansatz, one of our future tasks shall be 
practical computations of one- and two-loop Feynman diagrams, as~well as other next-to-leading and higher-order diagrams, for~studies of 
radiative corrections in lepton--hadron scattering processes. We hope that such a new method, along with the theoretical ones existing 
in the literature~\cite{Maximon:2000hm,Gramolin:2014pva,AGIM:2015,Bucoveanu:2018soy,Fadin:2018dwp,Banerjee:2020rww,Afanasev:2021nhy,Kaiser:2022pso}, 
will be useful for any  upcoming scattering experiments, such as PRad-II, Compass++/AMBER, PRES, MAGIX, and~ULQ2, that will be 
devoted to  measurements of the root-mean-square charge radius of the proton~\cite{Xiong:2023zih}.




\end{document}